\documentclass[10pt]{llncs}
\usepackage{etex}

\usepackage{times}

\usepackage{capt-of}
\usepackage{todo}
\usepackage{amssymb}
\usepackage{mathtools}

\usepackage{comment}

\usepackage[ruled,linesnumbered]{algorithm2e}
\SetAlFnt{\scriptsize}
\usepackage{algorithmic}

\usepackage{pgf}
\usepackage{tikz,subfigure}
\usetikzlibrary{arrows,automata}
\usetikzlibrary{positioning}
\usetikzlibrary{shapes,shadows}

\usepackage{xcolor}
\definecolor{sepia}         {cmyk}{0   , 0.83, 1   , 0.70}
\definecolor{royalblue}     {cmyk}{1   , 0.50, 0   , 0   }

\usepackage{pdfsync}

\usepackage{xtab}

\usepackage[pdfauthor={Hugo Daniel Macedo},%
 	    plainpages=false,%
	    pdfpagelabels,%
	    bookmarksnumbered,%
            colorlinks=true,%
            linkcolor=sepia,%
            citecolor=royalblue,%
            filecolor=sepia,%
            urlcolor=sepia,%
            pdftex,
            unicode,%
	]{hyperref}

\usepackage[T1]{fontenc}

\usepackage{amsmath}
\usepackage{amsfonts}

\usepackage{tikz}
\usetikzlibrary{chains,positioning,scopes}
\tikzset{node distance=2.5cm, auto}
\usepackage{tikz}
\usepackage{tikz-qtree}

\def\A{\mathcal{A}}
 
\def\C{\mathcal{C}}
\def\D{\mathcal{D}}
\def\F{\mathcal{F}}
\def\L{\mathcal{L}}
\def\N{\mathbb{N}}
\def\OutIn{\rightarrowtail}
\def\P{\mathcal{P}}
\def\PDS{\textbf{PDS}}
\def\MA{\textbf{MA}}
\def\TMA{\A^{\dagger}}
\def\T{\mathcal{T}}
\def\TAPIZFAP{\T(\F,\C)}
\def\TA{\mathbb{H}} 
\def\TS{\mathbb{TS}} 

\def\X{\mathcal{X}}   

\def\S{\textbf{$P_{API}$}}
\def\States{\textbf{States}}
\def\Reg{\textbf{R}}
\def\Z{\mathbb{Z}}
\def\EXP{\textbf{EXP}}

\def\pdt{\hookrightarrow}

\def\Conf#1#2{\left<#1,#2\right>}
\def\API{\textbf{API}}

\def\O{\mathcal{O}}
\def\JOIN{\mbox{\bf{insert\_subtree}}}
\def\subtree{\lhd}
\usepackage{wrapfig}

\def\CopyFile{\emph{CopyFile}}
\def\GetModuleFileName{\emph{GetModuleFileName}}

\def\GetWindowsDirectory{\emph{GetWindowsDirectory}}
\def\ShellExecute{\emph{ShellExecute}}

\def\ExitProcess{\emph{ExitProcess}}

\def\MSCDT{{\bf Mal\SCDT}}
\def\SCDT{\textbf{SCDT}}
\def\HELTA{\textbf{HELTA}}

\def\LdPinch{LdPinch}

\def\Nmalware{1176}
\def\Mtest{983}
\def\MtestF{330}
\def\Mtrain{193}
\def\Nbenware{250}
\def\FPR{0\%}

\author{Hugo Daniel Macedo \inst{1} \and Tayssir Touili \inst{1}}

\institute{LIAFA, CNRS and Univ. Paris Diderot, France\\ \email{\{macedo,touili\}@liafa.univ-paris-diderot.fr}}

\title{Mining malware specifications through static reachability analysis}

\begin{document}

\maketitle

\begin{abstract} 

The number of malicious software (malware) is growing out of control. Syntactic
signature based detection cannot cope with such growth and manual construction
of malware signature databases needs to be replaced by computer learning based
approaches. Currently, a single modern signature capturing the semantics of a
malicious behavior can be used to replace an arbitrarily large number of
old-fashioned syntactical signatures. However teaching computers to learn
such behaviors is a challenge. Existing work relies on dynamic analysis to
extract malicious behaviors, but such technique does not guarantee the coverage
of all behaviors. To sidestep this limitation we show how to learn malware
signatures using \emph{static} reachability analysis. The idea is to model
binary programs using pushdown systems (that can be used to model the stack 
operations occurring during the binary code execution), use reachability analysis to extract
behaviors in the form of trees, and use subtrees that are common among the
trees extracted from a training set of malware files as signatures. To detect
malware we propose to use a tree automaton to compactly store malicious
behavior trees and check if any of the subtrees extracted from the file under
analysis is malicious.  Experimental data shows that our approach can be used
to learn signatures from a training set of malware files and use them to detect 
a test set of malware that is 5 times the size of the training set.

\end{abstract}

\section{Introduction} 

Malware (malicious software) is software developed to damage the system that
executes it, e.g.: virus, trojans, rootkits, etc.   A malware variant performs
the same damage as another known malware, but its code, its syntactical
representation, is different. Malware can be grouped into families, sets of
malware sharing a common trait.  Security reports acknowledge a steady
increase in the number of new malware.  For instance, in 2010 the number of
newly unique variants of malware was 286 million \cite{Fossi:11} and recent
numbers confirm the trend \cite{MacAfee:12}.  Such numbers challenge current
malware detection technology and because variants can be automatically
generated the problem tends to get worse.  Research confirms the unsuitability
of current malware detectors \cite{Fredrikson:10,Song:12b}. The problem is the
low-level of the techniques used. 

The basic detection technique is signature matching, it consists in the inspection
of the binary code and search for patterns in the form of binary sequences
\cite{Szor:05}.  
Such patterns, malware signatures in the jargon and syntactic signatures
throughout this paper, are manually introduced in a database by experts. As it
is possible to automatically generate an unbounded number of variants, such
databases would have to grow arbitrarily, not to mention it takes about two months to manually update
them \cite{Fredrikson:10}.

An alternative to signature detection is dynamic analysis, which runs malware
in a virtual machine. Therefore, it is possible to check the program behavior,
for instance to detect calls to system functions or changes in sensitive
files, but as the execution duration must be limited in time it is difficult to
trigger the malicious behaviors, since these may be hidden behind user
interaction or require delays.

To overcome the problems of the previous techniques, a precise notion of
malicious behavior was introduced. Such is the outcome of the recent use of model-checking techniques to perform virus detection
\cite{Bergeron:01,Christodorescu:03,Christodorescu:05,Holzer:07,Kinder:05,Kinder:10,Song:12c,Song:12b,Song:12a,Singh:03}.
Such techniques allow to check the behavior (not the syntax) of the program
without executing it.
 A malicious behavior is  a pattern written
as a logical formula that specifies at a semantic level how the syntactic
instructions in the binary executable perform damage during execution.  As the 
malicious behavior is the same in all the variants of a malware, such patterns
can be used as modern (semantic) signatures which can be efficiently stored.


The prime example of a
malicious behavior is self-replication \cite{Szor:05}.  A typical
\begin{wrapfigure}{r}{0.4\textwidth}	
\vspace{-15pt}
\centering
\begin{footnotesize}
\begin{tabular}{l}
$l_1$ : push m			\\
$l_2$ : mov ebx 0		\\
$l_3$ : push ebx		\\
$l_4$ : call \GetModuleFileName	\\
$l_5$ : push m			\\
$l_6$ : call \CopyFile 		\\
\end{tabular}
\end{footnotesize}
  \vspace{-8pt}
\caption{Malware assembly fragment.}
  \vspace{-15pt}
\label{fig:dis}
\end{wrapfigure}
 instance of such behavior is a program
that  copies its own binary representation into another file, as 
exemplified in the assembly fragment of Fig. \ref{fig:dis}. The attacker
program discovers and stores  its file path into a memory address $m$ by calling the
\GetModuleFileName\ function with $0$ as first parameter and  $m$ as
second parameter. 
Later such file name is used to infect another file by calling \CopyFile\ with
$m$ as first parameter. 
Such malicious behaviors can naturally be defined in terms of system functions
calls and data flow relationships.  

\emph{System functions} are the mediators
between programs and their environment (user data, network access,\dots), and
as those functions can be given a fixed semantics, and are defined in an
Application Programming Interface (\API), they can be used as a common
denominator between programs, i.e.\ if the syntactical representation of programs 
is different but both interact in the same way with the environment, the programs 
are semantically equivalent from an observer perspective.  
   

A \emph{data flow} expresses that a value outputted at a certain time instant of
program execution by a function is used as an input by another function at a
following instant.  For example when a parameter is outputted by a system call
and is used as an input of another. Such data flow relations allow us to
characterize combined behaviors purported by the related system calls.  For
instance, in the example of Fig. \ref{fig:dis} it is the data flow evidenced by
the variable $m$, defined at the invocation of \GetModuleFileName\ and used at
the invocation of \CopyFile\ that establishes the self-replication behavior.

The malicious behaviors can be described naturally by trees expressing data
flows among system calls made at runtime. Due to code branches during execution 
  it is possible to have
several flows departing from the same system call, thus a tree structure is
particularly suitable to represent malicious behaviors.  
Plus, as such behaviors are described
independently of the functionality of the code that makes the calls, system
call data flow based signatures are  more robust against code obfuscations.
Thus, a remaining challenge is to learn such trees from malware
binary executables. 
%
%
 
Recent work \cite{Babic:11,Christodorescu:07,Fredrikson:10} shows that we
can teach computers to learn malicious behavior specifications. 
Given a set of malware, the problem of extracting
malicious behavior signatures consists in the extraction of the behaviors
included in the set and use statistical machinery to choose the ones that
are more likely to appear. However the approaches rely on dynamic 
analysis of executables which do not fully cover all behaviors.  
%
%
%
To overcome these limitations, in this paper we show how to use static
reachability analysis to extract malicious behaviors, thus covering the whole
behaviors of a program at once and within a limited time. 

\paragraph{\textbf{Our approach.}}
We address such challenge in the
following way: given the set of known malware binary executables, we extract
its malicious behaviors in the form 
of edge labeled trees with two kinds of nodes.  One kind represents the knowledge that
a system
function is called, the other kind of nodes represents which values were passed as parameters in the call (because some data
flows between functions are only malicious when the calls were made with a
specific parameter e.g. the $0$ passed to \GetModuleFileName\ in the
self-replication behavior).  Tree labels describe either a relation among system
calls or the number of the parameter instantiated. For example, the malicious
behavior displayed in Fig.  \ref{fig:dis} can be displayed in the tree shown in
Fig.  \ref{fig:treeex}. The tree captures  the self-replication behavior.

\begin{wrapfigure}{r}{0.35\textwidth}
\vspace{-25pt} \centering \begin{tikzpicture}[level distance=35pt,sibling
distance=18pt]

\Tree [ .$\GetModuleFileName$ 
	 \edge node[auto=right]{1}; [.$0$ ] 
	  \edge node[auto=left]{$2\OutIn1$}; [.$\CopyFile$ 
							]
         ]
\end{tikzpicture}
  \vspace{-8pt}
\caption{Self-replication behavior}
\label{fig:treeex}
\vspace{-15pt} \end{wrapfigure}
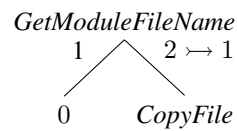   The
edge on the left means      
that the \GetModuleFileName\ function  is called with $0$ as first parameter (thus it will output the path to the malware file that called
it) while the edge on the right captures the data flow between the two system
calls i.e.\ the second parameter of a call to \GetModuleFileName\ is an output
and it is used as an input in the first parameter of a call to \CopyFile. Thus,
such tree describes the following behavior: \GetModuleFileName\ is called with
$0$ as first parameter and its second parameter will be used as input in the
first parameter of a subsequent call to \CopyFile.


The first step in the tree extraction process is to model the malware binaries,
which involves modeling (recursive) procedure calling and return, and parameter
passing that are implemented using a stack.  
For this aim, we model each of the files using a pushdown system (\PDS), an
automaton that mimics the binary code execution as a state transition system.
With this model one is able to rigorously define the behavior of the program
and use the decidable and efficient state reachability analysis of \PDS s to
calculate all the states and the contents of the stack that can occur during
execution. Therefore, if malware performs a system call with certain parameters, the
reachability analysis will reveal it even if the call is obfuscated, e.g.: jump to function address. The same happens if the call is
made using indirect addressing because the analysis will reveal that during execution the entry
point of the system call is reached.  Our approach also works against bitwise
manipulation of parameters, because we assume the system functions are not
changed by the attacker, thus when the executions reaches the entry point
of the system function, parameters must not be obfuscated, for instance in the example above even if 
the value of $m$ is obfuscated, at the entry point of the call 
the value must be $m$ to purport the self-replication behavior.


From the reachability analysis of each \PDS, we obtain a multi-automaton (\MA), a finite
automaton encoding the possibly infinite reachable configurations (states and
stack contents)\cite{Bouajjani:97,Esparza:00}. As the number of system
functions is finite, we cut the finite automaton to represent only the states
corresponding to system function entry points and stacks limited to the finite
number of parameters passed to the function. 

We analyze all data flows using the \MA s to build trees, written as system
call dependency trees (\SCDT s), representing such flows. The extracted trees
correspond to a superset of the data flows present in the malware because the
\PDS\ model is an overapproximation of the behaviors in the binary program.
This means, that when a data flow is found using our approach, there exists an
execution path in the model evidencing such data flow, but such execution path
may not be possible in the binary program due to approximation errors. 

From the trees (\SCDT s) extracted from the set of known malware binary
executables we use a data-mining algorithm to compute the most frequent
subtrees.  We assume such correspond to malicious behaviors and we will term
them malicious system call dependency trees (\MSCDT s). The usage of such
data-mining algorithm allows to compute behaviors, which we use as signatures
that are general and implementation details independent, therefore robust.

To store and recognize \MSCDT s we infer an automaton, termed \HELTA ,
recognizing trees containing \MSCDT s as subtrees. This allows to efficiently
store the malware signatures and recognize behaviors if they are hidden inside
another behavior.  The overview of the learning process from the malware files
to the database of semantic signatures is depicted in Figure \ref{fig:arch}. 

\tikzstyle{decision} = [starburst, draw, fill=red!20,text centered]
\tikzstyle{blocks} = [rectangle, draw, fill=blue!20, 
    text width=6em, text centered, rounded corners, minimum height=2em,double copy shadow={shadow 
xshift=0.1cm, shadow yshift=0.1cm}]
\tikzstyle{block} = [rectangle, draw, fill=blue!20, 
    text width=6em, text centered, rounded corners, minimum height=2em]
\tikzstyle{line} = [draw, -latex']
\tikzstyle{lines} = [draw,double copy shadow={shadow 
xshift=0.1cm, shadow yshift=0.1cm}]

\tikzstyle{arr} = [single arrow, draw,fill=blue!20]
\tikzstyle{arrs} = [single arrow, draw,fill=blue!20,double copy shadow={shadow 
xshift=0.1cm, shadow yshift=0.1cm}]

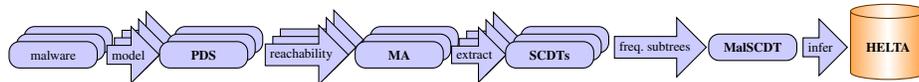
\begin{figure}
\centering
\vspace{-15pt} 

\begin{tiny}
\begin{tikzpicture}[node distance = 1.0cm,auto]
    \node [blocks] (init) {malware};
    \node [arrs,right of=init] (i2id) {model};
    \node [blocks, right of=i2id] (identify) {\PDS};
    \node [arrs,right of=identify,node distance=1.3cm] (id2ev) {reachability};
    \node [blocks, right of=id2ev,node distance=1.3cm] (evaluate) {\MA};
    \node [arrs,right of=evaluate] (ev2up) {extract};
    \node [blocks, right of=ev2up] (update) {\bf \SCDT s};
    \node [arr,right of=update,node distance=1.3cm,xshift=0.1cm,yshift=0.1cm] (u2st) {freq. subtrees};
    \node [block, right of=u2st,node distance=1.3cm] (stop) {\MSCDT};
    \node [arr,right of=stop,node distance=0.9cm] (st2a1) {infer};
	
    \node[draw, cylinder, alias=cyl, shape border rotate=90, aspect=1.6, %
    minimum height=1cm, minimum width=1cm, outer sep=-0.5\pgflinewidth, %
    color=orange!40!black, left color=orange!70, right color=orange!80, middle
    color=white,right of=st2a1,node distance=0.9cm] (a1) {};%
     \node[right of=st2a1,node distance=0.9cm] {\HELTA};%

\end{tikzpicture}
\end{tiny}

\vspace{-10pt} 
\caption{Learning malicious behaviors}
\label{fig:arch}
\vspace{-20pt} 
\end{figure}

To evaluate the efficiency of the computed malicious behaviors, we show they
can be applied to efficiently detect malware. To perform malware detection on a
binary executable, we extract trees using the same procedure used in the learning process
(described above), but applied to a single file.
%
We check whether the automaton storing malicious behaviors accepts any subtree
of the extracted trees (\SCDT s). If that is the case the executable contains a malicious
behavior and is classified as malware. The depiction of such process is shown in Figure \ref{fig:det}.

\begin{figure}
\centering
\vspace{-15pt} 
\begin{tiny}
\begin{tikzpicture}[node distance = 1.03cm,auto]
    \node [block] (init) {binary};
    \node [arr,right of=init] (i2id) {model};
    \node [block, right of=i2id] (identify) {\PDS};
    \node [arr,right of=identify,node distance=1.3cm] (id2ev) {reachability};
    \node [block, right of=id2ev,node distance=1.3cm] (evaluate) {\MA};
    \node [arr,right of=evaluate] (ev2up) {extract};
    \node [block, right of=ev2up] (update) {\bf \SCDT s};
    \node [arr,right of=update] (u2st) {subtrees};
    \node [decision, right of=u2st] (stop) {yes/no?};
    \node [arr,right of=stop,shape border rotate=180,node distance=1.3cm] (st2a1) {recognize};
	
    \node[draw, cylinder, alias=cyl, shape border rotate=90, aspect=1.6, %
    minimum height=1cm, minimum width=1cm, outer sep=-0.5\pgflinewidth, %
    color=orange!40!black, left color=orange!70, right color=orange!80, middle
    color=white,right of=st2a1,node distance=1cm] (a1) {};%
     \node[right of=st2a1,node distance=1cm] {\HELTA};%

\end{tikzpicture}
\end{tiny}
\vspace{-15pt} 
\caption{Malware detection}
\label{fig:det}
\vspace{-15pt} 
\end{figure}
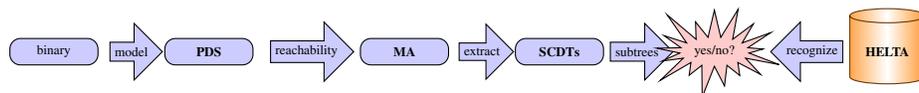

We implemented a tool that extracts the
behaviors and  selects the malicious candidates using  an algorithm for the
frequent subgraph problem\footnote{A tree is a special case of a graph}.  With
such tool we were able to infer some signatures not inferred using previous
approaches \cite{Babic:11,Christodorescu:07,Fredrikson:10} because our
signatures track calls to functions of the Win32 API instead of calls
to the Native API. It is a fact that it is always possible to use the previous
approaches to find Native API level signatures equivalent to the ones we infer,
therefore we do not claim our tool can express more behaviors, instead  we
claim that our approach is complementary to such works. It allows to express
behaviors at different API levels and to extract more abstract/readable (Win32
API level) signatures.
 
We obtained promising results, and we were able to detect \Mtest\ malware files
using the malicious trees inferred from \Mtrain\ malware files, with a \FPR\
false positive rate (thus showing our approach learns malicious behaviors that
do not appear in benign programs). This number of detected malware is larger
than the 16 files reported in \cite{Christodorescu:07} and in line with the 912
files detected in \cite{Fredrikson:10}. Our false positive detection rate is
better (5\% reported in \cite{Babic:11}).


\paragraph{Outline.} In Section \ref{sec:inf} we show how to model binary executables as \PDS s. 
Malware signatures are defined as labeled trees in Section \ref{sec:malspec}.
We present an algorithm  to infer malware specifications in Section \ref{sec:imp}, and 
we show how to use tree automata to perform malware detection in Section \ref{sec:det}. 
Experimental data shows our approach can be used to detect malware as detailed in Section \ref{sec:res}. The related work 
is summarized in Section \ref{sec:rel}
and in Section \ref{sec:conc} we present conclusions and future work.

\section{Binary code modeling}
\label{sec:inf}

Malware detection is performed directly in the executable encoding of the
software (binary code containing machine instructions and data). By modeling the
operational semantics of binary code, we are able to analyze it without relying
on execution. This section introduces the modeling framework  and how we 
 model executable files.


\subsection{Pushdown systems}

A pushdown system (\PDS) is a triple $\P=(P,\Gamma,\Delta)$ where $P$ is a finite
set of control points, $\Gamma$ is a finite alphabet of stack symbols, and
$\Delta\ \subseteq (P \times \Gamma) \times (P\times \Gamma^*)$ is a finite
set of transition rules.
A configuration $\Conf p \omega$ of $\P$ is an element of $P \times \Gamma^{*}$.
We write $\Conf{p}{\gamma} \pdt \Conf{q}{\omega}$ instead of $((p, \gamma), (q,
\omega)) \in \Delta$. The immediate successor relation $\leadsto_{\P} \subseteq (P \times
\Gamma^*) \times (P \times \Gamma^*)$ is defined as follows: if $\Conf{p}{\gamma}
\pdt \Conf{q}{\omega}$, then $\Conf{p}{\gamma \omega'} \leadsto_{\P} \Conf{q}{\omega \omega'}$ for
every $\omega' \in \Gamma^*$. 
%
%
%
%
%
The reachability relation
$\Rightarrow$ is defined as the reflexive and transitive closure of the
immediate successor relation.
%




Given a set of configurations $C$, $post(C)$ is defined as the set of
immediate successors of the elements in $C$. The reflexive and transitive
closure of $post$ is denoted as $post^*(C) = \{c' \in P\times\Gamma^* \mid
\exists c \in C, c \Rightarrow c' \}$ . Analogously  $pre(C)$ is defined as the
set of immediate predecessors of elements in $C$. Its reflexive and transitive
closure is denoted as $pre^*(C) = \{c \in P\times\Gamma^* \mid \exists c' \in
C, c \Rightarrow c' \}$.

\medskip 


Given a
pushdown system $\P=(P,\Gamma,\Delta)$, a $\P$-multi-automaton, $\P-MA$ or \MA\
when $\P$ is clear from context, is a tuple $\A=(\Gamma,Q,\delta,P,F)$, where $Q$ is a
finite set of states, $\delta \subseteq Q \times \Gamma \times Q$ is a
transition relation, $P \subseteq Q$ is the set of initial states corresponding
to the control points of $\P$, and $F\subseteq Q$ is a set of final states.
 
The transition relation for \MA\ is the smallest relation $\to \subseteq Q \times
\Gamma^* \times Q$ satisfying:
\begin{itemize}
\item $q \xrightarrow{\gamma} q'$ if $(q,\gamma,q') \in \delta$
\item $q \xrightarrow{\omega\gamma} q'$ if $q \xrightarrow{\omega}q''$ and $q''\xrightarrow{\gamma} q'$
\end{itemize}

$\A$ accepts (recognizes) a configuration $\Conf{p}{w}$  if $p \xrightarrow{w} q$ for some
$q\in F$. The set of configurations recognized by a \MA\  $\A$ is called regular and is designated by
$Conf(\A)$. The $post^*$ and $pre^*$ of regular configurations can be efficiently computed:

\begin{theorem}\cite{Bouajjani:97,Esparza:00} For a pushdown system $\P=(P,\Gamma,\Delta)$
and \MA\ $\A$, there exist \MA s $\A_{post^*}$ and $\A_{pre^*}$ recognizing
$post^*(Conf(\A))$ and $pre^*(Conf(\A))$  respectively. These can be constructed in polynomial time and space.   
\end{theorem}

\subsection{Modeling binary programs with \PDS s}



\def\Prog{\mathbb{P}}

We use the approach detailed in \cite[Section 2]{Song:12b} to model each executable program $\Prog$. The approach relies on the assumption that there
exists an oracle $\mathcal{O}$ computing  a \PDS\
$\P=(P,\Gamma,\Delta)$ from the binary program, where $P$ corresponds to the
control points of the program, $\Gamma$ corresponds to the approximate set of
values pushed to the stack, and $\Delta$ models the different
instructions of the program.  The obtained \PDS\ mimics the runs of program $\Prog$. 

In addition to the approach of \cite{Song:12b}, let  $\API$ be the set of all Application Programming Interface function names
available in the program.  We assume the oracle $\mathcal{O}$ approximates
the set $\S \subseteq P$ of control points of a program that correspond to
instruction addresses that at program runtime are translated (dynamically linked) by the operating system into
system function entry points, the number of parameters of such functions and the type of each
parameter.  We consider a simple type system:
$
\tau ::= \mbox{\emph{in}} \mid \mbox{\emph{out}}$ (\emph{in} for input
parameter, and \emph{out} for output) containing the atomic value \emph{out}
used to denote a parameter that is modified after function execution and
\emph{in} to denote the parameter is not changed by the function.

We assume, $\mathcal{O}$  computes a function 
$\varrho_{\lambda} : \S \to \API$ that identifies program control points 
corresponding to system calls with an unique function name, a function $\varrho_\tau
: \S \times \N \to 2^\tau$ such that $\varrho_\tau(p,n)$ is the set\footnote{The
API defines parameters that are both input and
output.} of possible types of the $n$-th parameter of the
system call that has $p$ as entry point, and a function $\varrho_{ar} : \S \to \N$
defining the number of parameters for each system call in $\S$ 
%
For example, if we consider the program of Fig. \ref{fig:dis}, we obtain  $\S=\{l_g,l_c\}$ since these two points correspond to system call entry points, $\varrho_{\lambda}(l_g)=GetModuleFileName$ since $l_g$ corresponds to the entry point of the function \GetModuleFileName.
 $\varrho_{ar}(l_g) = 3$ since \GetModuleFileName\ has three parameters, and
$\varrho_\tau(l_g,2) = \{ out \}$ since the second parameter of the \GetModuleFileName\ function is 
defined as an output, and analogously $\varrho_\tau(l_g,1) =
\varrho_\tau(l_g,3)= \{in\}$, since these correspond to input parameters.


\section{Malicious behavior specifications}
\label{sec:malspec}

As already mentioned, malicious behaviors, data flow relationships between
system function calls, will be expressed as trees where nodes represent system
functions or parameter values and edges specify the data flow or the number of the parameter
to which the value was passed. We will now formally introduce
the notion of edge labeled trees.

\subsection{Edge labeled trees}

An unranked alphabet is a finite set $\F$ of symbols. 
Given an unranked alphabet $\F$, let a set of colors $\C$ be an alphabet of unary symbols
and disjoint from $\F$, and $\X$ be a set of variables disjoint from $\F$.  The set $\T(\F,\C,\X)$ of colored terms over the unranked
alphabet $\F$, colors $\C$ and variables $\X$  it is the smallest set of terms such that:

\begin{itemize}
\item $\F \subseteq \T(\F,\C,\X)$,
\item $\X  \subseteq \T(\F,\C,\X)$, and
\item $f(c_1(t_1),\dots,c_n(t_n)) \in \T(\F,\C,\X)$, for $n\geq 1$, $c_i \in \C$, $t_i \in \T(\F,\C,\X)$.
\end{itemize}  

\begin{wrapfigure}{r}{0.2\textwidth}
\vspace{-20pt}
\centering
\begin{tikzpicture}[level distance=30pt,sibling distance=18pt]

\Tree [ .$f$ 
	 \edge node[auto=right]{$c_1$}; [.$a$ ] 
	 \edge node[auto=left] {$c_2$}; [.$b$ ] 
      ]
\end{tikzpicture}
\vspace{-5pt}
\caption{Example}
\label{fig:edg}
\vspace{-20pt}
\end{wrapfigure} 
If $\X = \emptyset$ then $\T(\F,\C,\X)$ is written as $\T(\F,\C)$, and its elements are designated as ground terms.
Each element of the set of  terms  can be represented by an edge labeled tree.
For example, let $\F=\{f\}$, $\C=\{c_1,c_2\}$, and $\X=\emptyset$. The colored tree
$f(c_1(a),c_2(b)) \in \T(\F,\C)$ can be represented by the edge labeled tree of Fig.
\ref{fig:edg}.

Let $\X_n$ be a set of $n$ variables. A term $E \in \T(\F,\C,\X_n)$ is called an
environment and the expression $E[t_1,\dots,t_n]$ for $t_1,\dots,t_n \in \T(\F,\C)$
denotes the term in $\T(\F,\C)$ obtained from $E$ by replacing the variable $x_i$
by $t_i$ for each $1 \leq i \leq n$. 

A subtree $t'$ of a  tree $t$ in $\T(\L,\C)$, written as $t' \subtree t$,  is a
term such that there exists an environment $E$ in $\T(\L,\C,\{x\})$ where
$x$ appears only once and $t = E[t']$.

The tree $f(c_1(a),c_2(b))$ represents the same behavior as tree
$f(c_2(b),c_1(a))$. Thus, to  efficiently compare edge labeled trees, and to
avoid missing malicious behaviors due to tree representation, we define a
canonical representation of edge labeled trees. We assume that $\F$ and $\C$
are totally ordered.

A term is in canonical form if it is a constant (leaf) or if it is a function
(tree node) where each argument is in canonical form and arguments are sorted
without repetitions by term order.

Let $c \in \C$ and $t\in \F(\C,\T)$ such that $\F, \C,$ and $\T$ are respectively ordered by $<^\F, <^\C,$ and $<^\T$, and $t$ is  in canonical form. We assume a subtree insertion operation $(\JOIN)$ 
where $\JOIN(c(t),t')$ adds $c(t)$ as a child to the root of $t'$ in the correct place to maintain a canonical representation of the tree, overwriting if the subtree $c(t)$ already exists.

\subsection{System call dependency trees}

We will represent malware behaviors as trees encoding data flow relationships
between system function calls. 
Tree nodes represent either system functions or
parameter values. Edge colors label the characteristics of the data flow
between functions, e.g. $2 \OutIn 1$ labeling an edge from function $f$ and
$f'$ means that at some point $f$ is called with some value $v$ as second
parameter, which is of type $out$, and afterwards $f'$ is called with $v$ as
first parameter, which in turn is of type $in$. Moreover, when an edge connects
a node labeled with function $f$ and a child node with some value $v$, meaning
the function was called with parameter $v$, it will be labeled with  the number
of the parameter, thus to represent a call was made with $0$ as first parameter
to function $f$, we add $1$ as a label of the edge from node $f$ to node  $0$.

\begin{definition}
 Formally,
let $\F$ be the set of all system call function names (the union of all possibly  $\API$ function names returned by the oracle of Section \ref{sec:inf}) and values passed as function parameters (a subset of the union of all $\Gamma$ sets calculated by the oracle).
In addition, let $\C$ be a set of colors containing all the possible parameter numbers and data flows, i.e.: 
$
\C=\{1,\dots,\max_{f \in \API}(\varrho_{ar}(f))\} \cup \{ x \OutIn y \mid x,y \in \{1,\dots,\max_{f\in \API}(\varrho_{ar}(f))\}\}$
A System Call Dependency Tree, written as \SCDT, is defined as a ground term of the set $\T(\F,\C)$. 
\end{definition}

\noindent{\textbf{Example.}} Let 
$\F=\{0,\GetModuleFileName,\CopyFile\}$ and $\C=\{1,2\OutIn
1\}$, the behavior of Fig. \ref{fig:treeex} can be described by 
$t=\GetModuleFileName\big(1(0),2\OutIn 1(\CopyFile)\big)$.

\section{Mining malware specifications}
\label{sec:imp}
\def\PPDS{\mathcal{P}}

In this section we show how to compute the \SCDT s corresponding to malware
behaviors that we will use as malware specifications.  Given a finite set of
programs $\Prog_1,\dots, \Prog_q$ known to be malicious in advance we compute
\PDS s $\PPDS_1,\dots,\PPDS_q$ that model these malicious programs.  Then, for
each \PDS\ $\PPDS_i$ we compute a set of trees $\TS_i$ that contains the data
flows represented as \SCDT s for the program $\Prog_i$.  From the computed set of
trees for each program, $\TS_1,\dots, \TS_q$, we calculate the common subtrees,
the ones that are most probable to appear in malware, that we use as malware
specifications.

To compute the sets of trees $\TS_i$ we proceed as follows: For each program
$\Prog_i$ modeled as a \PDS\  $\PPDS_i$  we compute the finite automaton
encoding the set of reachable configurations from the initial state using the
reachability analysis algorithm from \cite{Esparza:00}. As there may be an
infinite number of configurations and we are only interested in the
configurations whose control points correspond to a system function entry
point with some finite number of elements in the stack (only the parameters of
the function under consideration are important), we build another automaton
recognizing such finite set of configurations. For each of such configurations,
understood as possible data flow origins,  we repeat the process to calculate
the reachable configurations, understood as possible data flow destinations.
Then, if a data flow between configurations is found, i.e.\ the value passed as a
parameter to an origin configuration has type $out$ and the same value passed as a
parameter of type $in$ to a destination configuration, we build a \SCDT\ with
the origin function as root node and an edge to a node corresponding to the
 destination function.

To calculate the common subtrees we use the algorithm \cite{Yan:02} 
computing frequent subgraph, to compute frequent subtrees.


\subsection{System call targeted reachability analysis}
\label{sec:comp}

To compute the data flows for a malware pushdown system  model
$\P=(P,\Gamma,\Delta)$, we first calculate the reachability of $\P$ using the
 algorithms presented in \cite{Esparza:00}.  From $\P$ we
build the (\MA) automaton $\A$ that recognizes the $post^{*}(\Conf{p_i}{\epsilon})$,
i.e.\ the  set of reachable configurations from the initial configuration
$\Conf{p_i}{\epsilon}$, where $p_i$ is a designated initial control point and
$\epsilon$ denotes the empty stack.
\paragraph{\textbf{\MA\ Trimming.}} To compute data flows between system call
related control points  $p_o,p_d \in \S$ with parameter numbers
$\varrho_{ar}(p_o)=m$ and $\varrho_{ar}(p_d)=n$ we need to consider only the
top $m+1$ and $n+1$ elements of the stack reached at control points $p_o$ and
$p_d$ because, in assembly, parameters are passed to functions through the
stack. Before invoking a function the parameters are pushed in reverse order
into the stack, and after the return address is pushed. Thus, if a function
receives $m$ parameters, then at its entry point, for instance $p_o$, the top $m+1$ elements of the
stack correspond to the parameters plus the return address.
Thus we only need to consider the top $m+1$ elements of the stack reached at control point $p_o$.
This is the reason why we can analyze the possibly infinite number of configurations encoded in the reachability resulting 
finite automaton, we only inspect a finite subset.
To abbreviate the algorithm that computes \SCDT\ we define such subset of configurations 
in terms of a new automaton obtained by cutting the \MA\ resulting from the reachability analysis. 
\begin{definition}
Given a \MA\ $\A$ recognizing the reachable configurations of a \PDS\ $\P=(P,\Gamma,\Delta)$  
we define the trim automaton $\TMA$ as the automaton
recognizing the configurations in the set: $\{\Conf{p}{w} \in \S \times \Gamma^{*} \mid |w| = \varrho_{ar}(p)+1
\wedge \exists w' \in \Gamma^* \mbox{s.t.} \Conf{p}{ww'} \allowbreak \mbox{is accepted by }
\A\}$

\end{definition}

Intuitively, we cut the automaton and keep only configurations where 
control points $p$ correspond to system function entry points, and 
the stacks are bounded by the number of parameters of the function plus
one to take into account the return address. 
The trim operation will be written as $\Psi$, thus  $\TMA=\Psi(\A)$.  It is
trivial to prove that the $Conf(\TMA)$ is a finite language, in fact the number
of configurations corresponding to valid system call function entry point, and
its finite number of parameters is at most:\\
 $O(|\S|\cdot|\Gamma|\cdot
\operatorname*{max}_{p\in \S}(\varrho_{ar}(p))).$
\vspace{-0.4cm}
\subsection{Extracting \SCDT s} \label{sec:computing} Algorithms \ref{alg:extract}
and \ref{alg:build} detail our approach to extract behaviors.  We assume a
maximum tree height $h\in \N$ is given as input.  We write $\omega[n]$ to
denote the $n$-th element of some word $\omega \in \Gamma^*$.
%

%
\begin{wrapfigure}{r}{0.55\textwidth}
\vspace{-1.25cm}
\begin{algorithm}[H]
\SetAlgoLined
\LinesNumbered
\SetKwFunction{BuildSCDT}{BuildSCDT}
\SetKwFunction{Trees}{$\TS$}


\nllabel{alga:l2} \ForAll{$\P_i$}{
\nllabel{alga:l1}  	$\Trees_i \longleftarrow \emptyset$\;

\nllabel{alga:l3}	$\TMA_i \longleftarrow \Psi(post^*(\Conf{p_i}{\epsilon}))$\;

\nllabel{alga:l4}	\ForAll{$\Conf{p_o}{\omega_o}  \in Conf(\TMA_i)$}
			{
\nllabel{alga:l5}
		$\Trees_i \longleftarrow \Trees_i \cup \{\BuildSCDT{$\Conf{p_o}{\omega_o}$,$h$}\}$\;
\nllabel{alga:l8}
			}
 				}

	
 
 \Return{$\Trees$}\;
 \caption{ExtractSCDT}
\label{alg:extract}
\end{algorithm}
\vspace{-0.8cm}
\end{wrapfigure}

The Algorithm \ref{alg:extract} iterates over the models
$\PPDS_1,\dots,\PPDS_q$ (line \ref{alga:l2}). For each it initializes the set
of resulting trees to the empty set (line \ref{alga:l3}) and computes the
configurations corresponding to system calls that are reachable from the given
initial configuration $\Conf{p_i}{\epsilon}$ (line \ref{alga:l4}).  The initial
configuration is built using the binary executable entry point and an empty
stack. Then, for every configuration corresponding to a system call entry point
$\Conf{p_o}{\omega_o}$  recognized by the trim automaton  (line \ref{alga:l5})
it calls \textbf{BuildSCDT} to build a \SCDT\ tree of height at most $h$ with
the function of entry point $p_o$ as root (line \ref{alga:l8}). 
%

%
%
The \textbf{BuildSCDT} procedure is displayed in Algorithm \ref{alg:build}, it
is used to recursively build a tree. First, the tree to be returned is
initialized to be the origin system call entry point $p_o$ (line \ref{algb:l4}).
When the maximum desired tree height is not reached (line \ref{algb:l12}),
we calculate what are the system calls reached from $\Conf{p_o}{\omega_o}$ (line \ref{algb:l3})
and check for flows to any system call related configuration $\Conf{p_d}{\omega_d}$ (line \ref{algb:l5}).
If a data flow is found between two configurations (line \ref{algb:l7}), i.e.\
there are parameter numbers $n$ and $m$ such that the value passed to system
call at control point $p_o$ is the same as the value
passed in position $m$ of system call at a control point $p_d$, and there is in fact a
flow (line \ref{algb:l8}) i.e.\ the parameter $n$ of the function corresponding to the entry point $p_o$ is of type \emph{out}  and the parameter $m$ of the function corresponding to the entry point $p_d$ is of type \emph{in}, 
we add a new child  with label $n \OutIn m$ to the recursively
computed tree for the destination system call $p_d$ (line \ref{algb:l9}).

\begin{algorithm}[H]
\SetKwFunction{BuildSCDT}{BuildSCDT}
\SetKwFunction{Join}{$ \JOIN\;$}
\SetKwData{Tree}{tree}
\SetAlgoLined
\LinesNumbered
				$\Tree = \varrho_{\lambda}(p_o)$\;\nllabel{algb:l4}
	\If{$h > 0$}{\nllabel{algb:l12}

	$\TMA \longleftarrow \Psi(post^*(\Conf{p_o}{\omega_o}))$\;\nllabel{algb:l3}
		\ForAll{$\Conf{p_d}{\omega_d} \in Conf(\TMA)\setminus \{\Conf{p_o}{\omega_o}\}$}
				{\nllabel{algb:l5}

\nllabel{algb:l6} \ForAll{ $(n,m)$ s.t. $1\leq n \leq \varrho_{ar}(p_o) \wedge 1\leq m\leq \varrho_{ar}(p_d)$}
							{	
\nllabel{algb:l7}			\If{$w_o[n]=w_d[m] \wedge \varrho_\tau(p_o,n)=out \wedge \varrho_\tau(p_d,m)=in$} 
					{

\nllabel{algb:l8}				$\Tree \longleftarrow  \Join (n\OutIn m(\BuildSCDT(\Conf{p_d}{\omega_d},h-1)),\mbox{\Tree}) $\;
\nllabel{algb:l9}					}

							}

				}
	}
	{

					\ForAll{$n \in \{1,..,\varrho_{ar}(p_o)\}$}
							{	

\nllabel{algb:l10}				$\Tree \longleftarrow  \Join(n(w_o[n]),\mbox{\Tree}) $\;
\nllabel{algb:l11}					}

	}						
				
\Return $\Tree$\;
\caption{BuildSCDT}
\label{alg:build}
\end{algorithm}

To add the edges representing the values passed as parameters in the call of $p_o$ we iterate over the possible number of parameters
of the origin system call entry point (line \ref{algb:l10}) and add an edge
with the number of parameter $n$ and the value passed in the stack  $\omega_o[n]$ (line
\ref{algb:l11}).
When the maximum desired tree height is reached, the algorithm returns only 
a tree with $p_o$ as root and the values passed as parameters in the call.

\subsection{Computing malicious behavior trees}
\label{sec:min}


After extracting \SCDT s  for each of the inputed malware programs,
one has to compute which are the ones that correspond to malicious behaviors.
The \SCDT s that correspond to malicious behaviors will be named malicious trees.
To choose the malicious trees  we compute the most frequent subtrees 
in the set $\TS$ of trees extracted from the  set of malware used to train our detector. 
For that we need the notion of \emph{support set}, the set of trees containing some given subtree,
and the notion of \emph{tree support} that gives the ratio of trees containing the subtree to the whole set of trees.

Given a finite set of trees $\TS \subseteq \TAPIZFAP$ and a tree $t\in\TS$, the \emph{support set} of a
tree $t$ is defined as  $T_t = \{t' \mid t \subtree t', t' \in \TS\}$. The
\emph{tree support} of a tree $t$ in the set $\TS$ is calculated as $sup(t) = \frac{|T_t|}{|\TS|}$. For a fixed
threshold $k$ the set of frequent trees of $T$ is the set of trees with
\emph{tree support} greater than $k$.  
\begin{definition} 
For a set of system call dependency trees trees $\TS\subseteq \TAPIZFAP$ and 
a given threshold $k$, a malicious behavior tree is a tree $t \in \TS$ s.t. 
$sup(t) \geq k$. The set of malicious behavior trees will be called \MSCDT. 
\end{definition}


To compute frequent subtrees we specialize the frequent subgraph algorithm presented in 
\cite{Yan:02} to the case of trees. The algorithm receives a set of trees and a 
support value $k\in[0,1]$ and outputs all the subtrees with support at least $k$.
The graph algorithm works by defining a lexicographical
order among the trees and mapping each to a canonical representation using a 
code based on the depth-first search tree generated by the traversal. Using
such lexicographical order the subtree search space can be efficiently explored avoiding 
duplicate computations.

\section{Malware detection}
\label{sec:det}

We show in this section how the malicious behaviors trees that we computed
using our techniques can be used to efficiently detect malware.
%
To decide whether a given program $\Prog$ is malware or not, we apply again the
technique described in Section \ref{sec:imp} to compute the \SCDT s for the
program $\Prog$ being analyzed. Then we check whether such trees correspond to
malicious behaviors, i.e.\ whether such trees contain subtrees that correspond
to malicious behaviors.

\begin{wrapfigure}{r}{0.4\textwidth}	
\vspace{-25pt}
\centering
\begin{tikzpicture}[level distance=35pt,sibling distance=16pt]

\Tree [ .$\GetModuleFileName$ 
	 \edge node[auto=right]{1}; [.$0$ ] 
	  \edge node[auto=left]{$2\OutIn1$}; [.$\CopyFile$ 
							]
	  \edge node[auto=left]{$1\OutIn1$}; [.$\ExitProcess$ 
							]
         ]
\end{tikzpicture}
\vspace{-8pt}
\caption{Behaviors extracted from $\Prog$}
\label{fig:spur}
\vspace{-15pt}
\end{wrapfigure}
To efficiently perform this task, we use tree automata. The
advantage of using tree automata is that we can build the minimal automaton
that recognizes the set of malicious signatures, to obtain a compact and efficient
database. Plus, malware detection, using membership in automata, can be done efficiently.
However, we need to adapt tree automata to suite malware detection, that is, to
define automata that can recognize {\em edge labeled} trees. Furthermore, we
cannot use standard tree automata because the trees that can be generated from
the program $\Prog$ to be analyzed may have arbitrary arities (since we do not know
a priori the behaviors of $\Prog$).  For example the behavior of the program $\Prog$ can be
described by the tree of Fig. \ref{fig:spur} that contains the
self-replication malicious behavior of Fig. \ref{fig:treeex}. However, if we
use a binary tree automaton $\TA$ to recognize the tree of Fig.
\ref{fig:treeex}, $\TA$ will not recognize the tree of Fig. \ref{fig:spur}, because $\Prog$ contains the malicious behaviors and 
extra behaviors. To
overcome this problem we will use unranked tree automata (a.k.a. hedge
automata), since the trees that can be obtained by analysing program $\Prog$ might
have arbitrary arity.

In this section, we show how to use hedge automata for malware detection.
First, we give the formal definition of hedge automata. Then, we show how we
can infer a hedge automaton to recognize malicious behaviors that may be contained in some tree. And we conclude
by explaining how to use it to detect malware.

\subsection{Tree automata for edge labeled trees}
\def\Q{Q^\TA}
\def\QHELTA{\mathcal{A}}
\def\DeltaTA{\Delta^\TA}
\def\toHA{\xrightarrow{}_{\TA}}
\def\toHAS{\xrightarrow{*}_{\TA}}
\begin{definition}	
An \emph{hedge edge labeled tree automaton} (\HELTA) over $\T(\F,\C)$ is a tuple
$\TA=(\Q,\F,\C,\QHELTA,\DeltaTA)$ where $\Q$ is a finite  set of states, $\QHELTA \subseteq \Q$ is
the set of final states, and $\DeltaTA$ is a finite set of rewriting rules
defined as $f(R)
       \to 
       q$ for $f \in \F$, $q \in \Q$,  and $R \subseteq \left[\C(\Q)\right]^*$ is a regular word language over $\C(\Q)$ i.e.\ the language encoding all the possible children of the tree node $f$.

We define a move relation $\toHA$ between ground terms in $\T(\F\cup \Q,\C)$ 
as follows: 
Let $t ,t' \in \T (\F \cup \Q,\C)$, the move relation $\toHA$ is defined by: 
$t \toHA t'$ iff there exists an environment $E \in \T(\F\cup \Q,\C,\{x\})$, a rule $r= f(R) \to q \in \DeltaTA$
such that $t = E[f(c_1(q_1),\dots, c_n(q_n)))]$, and $c_1(q_1) \dots c_n(q_n)\in R$, and  $t'= E[q]$.
We write $\toHAS$ to denote the reflexive and transitive closure of $\toHA$.
Given an \HELTA\  $\TA=(\Q,\F,\C,\QHELTA,\DeltaTA)$ and an edge labeled tree $t$,
we say that $t$ is accepted by a state $q$ if $t \toHAS q$, $t$ is accepted by $\TA$
if $\exists q \in \QHELTA$  s.t. $t \toHAS q$.
\end{definition}

Intuitively, given an input term $t$, a run of $\TA$ on $t$ according to the move
relation $\toHA$ can be done in a bottom-up manner as follows: first, we assign
nondeterministically a state $q$ to each leaf labeled with symbol $f$ if there is
in $\DeltaTA$ a rule of the form $f(R) \to q$ such that $\epsilon \in R$. Then, for each node labeled with
a symbol $f$, and having the terms $c_1(t_1), \dots , c_1(t_n)$ as children, we must collect
the states $q_1, \dots , q_n$ assigned to all its children, i.e., such that $c_i(t_i) \toHAS q_i$
, for $1 \leq i \leq n$, and then associate a state $q$ to 
the node itself if there exists in $\DeltaTA$ a rule $r=f(R)\to q$ such that $q_1 \dots q_n \in R$.
A term $t$ is accepted if $\TA$ reaches the root of $t$ in a final state.

\subsection{Inferring tree automata from malicious behavior trees}
\label{sec:infta}
In this section we show how to infer an
\HELTA\ recognizing trees containing the inferred malicious behaviors. 
Thus, if $t$ is a malicious behavior, and $t'$ is a behavior of a program $\Prog$ 
that is under analysis, such that $t'$ contains the behavior described by $t$, the automaton 
must recognize it.
As an example assume
$t \in \MSCDT$ is a tree of the form $f(c_1(a),c_2(b)))$, s.t. $a,b \in \F$  and $E \in \T(\F,\C,\{x\})$ is an environment, 
then the automaton must recognize trees $t'$ of the form:
\begin{small}
$E[\mathbf{f}(
c^{1}_1(t^{1}_1),\dots,c^{1}_{m_1}(t^{1}_{m_1}),\mathbf{c_1}(\mathbf{a}(e_1)),\allowbreak c^{2}_1(t^{2}_1),\allowbreak \dots,c^{2}_{m_2}(t^{2}_{m_2}),
   \mathbf{c_2}(\mathbf{b}(e_2)),c^{3}_1(t^{3}_1),\dots,c^{3}_{m_{3}}(t^{3}_{m_{3}}) 
)]$
\end{small} 
meaning the tree is embedded in other tree, i.e. $t$ is a subtree of $t'$  and it may have extra behaviors $c_i^j(t_i^j)$ and also extra subtrees $e_1, e_2 \in \T(\F,C)$ as child of the leafs $a$ and $b$.

Let $t \in \MSCDT$, we define the operation $\Omega : \MSCDT \to \TAPIZFAP$ that transforms a malicious tree into  the set of all system call dependency trees
containing the malicious behavior $t$. $\Omega$ is defined inductively as:
\begin{itemize}
\item[(1)] $\Omega(a) = \{a(t) \mid t \in \TAPIZFAP \}$, if $a \in \F$ is a leaf,
\item[(2)] $\Omega(f(c_1(t_1),\dots,c_n(t_n))) = \{f(
c^{1}_1(t^{1}_1),\dots,c^{1}_{n_1}(t^{1}_{n_1}),c_1(\Omega(t_1)),c^{2}_1(t^{2}_1),\allowbreak \dots \allowbreak , \qquad \allowbreak c^{2}_{n_2}(t^{2}_{n_2}),\allowbreak \dots,
  c^{n}_1(t^{n}_1),\dots,c^{n}_{n_n}(t^{n}_{n_n}),c_n(\Omega(t_n)),c^{n+1}_1(t^{n+1}_1),\dots, \allowbreak c^{n+1}_{n_{n+1}}(t^{n+1}_{n_{n+1}})
  )   \mid \qquad \allowbreak  c_i^{j} \in \C \allowbreak \mbox{ and } t_i^{j} \in \TAPIZFAP \}$, otherwise.
\end{itemize}
The first rule asserts that after the leaves of the malicious behavior $t$ there may be other behaviors, while the second asserts that in the nodes of the 
tree $t'$ there may be extra behaviors, for instance the edge to $\ExitProcess$ in Fig. \ref{fig:spur}. Then, if $t$ is a
malicious behavior tree, we would like to compute an \HELTA\ that recognizes all the trees $t'$ s.t. $\exists t'' \in \Omega(t)$ and $t'=E[t'']$ for an environment $E \in \T(\F,\C,\{x\})$.      

Let $\MSCDT$ be a finite set of malicious trees, by definition each $t \in \MSCDT$ is a term of
$\TAPIZFAP$. We infer an \HELTA\ $\TA=(\Q,\F,\C,\QHELTA,\DeltaTA)$
recognizing trees containing malicious behaviors. Where $\Q = \{ q_t \mid t \subtree t' \mbox{ and } t' \in \MSCDT \} \cup \{ q_t \mid t \in \F \}$
 i.e. contains a state for
each subtree of the trees to accept, plus a state for each possible symbol of the alphabet 
that will be reached when a subtree with such symbol as root is not recognized. 
 The final states are defined as the states that correspond to recognizing a malicious tree $\QHELTA = \{ q_t \mid  t \in \MSCDT \}$. And $\DeltaTA$ is defined by rules:
\begin{enumerate}
\item[R1] For all $f\in \F$,  $
f([\C(\Q)]^*) \to q_{f} \in \DeltaTA
$
\item[R2] For all  $t=f(c_1(t_1),\dots,c_n(t_n))$ such that $t \subtree t'$ and $t' \in \MSCDT$,
\begin{small}
$f(\left[\C(\Q)\right]^* \allowbreak c_1(q_{t_1})\left[\C(\Q)\right]^*\dots\left[\C(\Q)\right]^*c_n(q_{t_n})\left[\C(\Q)\right]^*) \to q_{f(c_1(t_1),\dots,c_n(t_n))} \in \DeltaTA
$\end{small}
\item[R3] For all final state $q_t \in \QHELTA$ and all $f\in \F$,
$
f(\left[\C(\Q)\right]^*,q_t,\left[\C(\Q)\right]^*) \to q_t \in \DeltaTA
$
\end{enumerate}

Intuitively, for $f \in \F$, states $q_f$ recognize all the terms whose roots are $f$. This is ensured by R1. In the rules 
$[\C(\Q)]^*$ allows to recognize terms $t$ in (1) and $c_i^j(t_i^j)$ in (2). For a subtree $t_i$ 
of a malicious behavior $t$ in every \MSCDT, $q_{t_i}$ recognizes $\Omega(q_{t_i})$. This is ensured by rules R2, which 
 guarantees that a malicious tree containing extra behaviors is recognized. R3 guarantees that a tree containing a malicious behavior as subtree is recognized, i.e. R3 
ensures that if $t$ is a malicious behavior and $E\in\T(\F,\C,\{x\})$ is an environment, then $q_t$ recognizes $E[t']$ for
every $t'$ in $\Omega(t)$.

In the following we assert that if a tree $t'$ contains a subtree $t''$ that
contains a malicious behavior $t$, then the inferred automaton will recognize it
(even if there are extra behaviors). Proof should follow by induction.

\begin{theorem}
Given a term $t \in \MSCDT$, and $t' \in \TAPIZFAP$. If there $\exists t'' \in \Omega(t)$ and an environment $E\in\T(\F,\C,\{x\})$ and $t'=E[t'']$, then $t' \toHAS q_t$. 
\label{thm:ta}
\end{theorem}

%
%

\subsection{Malware detection}

The detection phase works as
follows. Given a program $\Prog$ to analyze we build a \PDS\ model $\P$ using the
approach described in Section \ref{sec:inf}, then we extract the set of
behaviors $\TS$ contained in $\Prog$ using the approach in Section \ref{sec:imp}.
 Then we use the automaton $\TA$ to search if any of the trees in $\TS$
can be matched by the automaton.  If that is the case the program $\Prog$ is deemed
malware.  

\paragraph{\bf Example.} Suppose the tree in Fig. \ref{fig:spur} was extracted and the  tree in Fig. \ref{fig:treeex} is the only malicious behavior in \MSCDT, which in turn is defined using
$\C=\{1,2 \OutIn 1\}$
and
$\F=\{0,\CopyFile,\ExitProcess,\GetModuleFileName\}$. We  define an automaton $\TA$ where the set of states is
$
\Q = \{q_0, \allowbreak q_{\ExitProcess},q_{\CopyFile},q_{\GetModuleFileName(1(0),2 \OutIn 1 (\CopyFile))}\}
$, the accepting set is
$
\QHELTA = \{q_{\GetModuleFileName(1(0), \allowbreak 2 \OutIn 1 (\CopyFile))}\}
$, and $\DeltaTA$ contains rules processing the leaves: 
$
0([\C(\Q)]^*) \toHA q_{0}
$, 
$
\ExitProcess([\C(\Q)]^*) \toHA q_{\ExitProcess}
$, and 
$
\CopyFile([\C(\Q)]^*) \allowbreak \toHA q_{\CopyFile}
$. And a rule 
$
\GetModuleFileName(\allowbreak [\C(\Q)]^*,1(q_{0}), \allowbreak[\C(\Q)]^*, \allowbreak 2\OutIn1(q_{\CopyFile}), \allowbreak [\C(\Q)]^*)  \toHA  q_{\GetModuleFileName(1(0),2\OutIn1(\CopyFile))}
$ processing the whole malicious behavior of Fig. \ref{fig:treeex}. 


\section{Experiments}
\label{sec:res}


To evaluate our approach, we implemented a tool prototype that was tested on a
dataset of real malware and benign programs. The input dataset of malware
contains \Nmalware\ malware instances (Virus, Backdoors, Trojans, Worms,\dots)
collected from virus repositories as VX Heavens and a disjoint dataset of
\Nbenware\ benign files collected from a Windows XP fresh operating system
installation. 
We arbitrarily split the malware dataset into a training and test
group.  The train dataset was used to infer the malicious trees that were used
in the detection of the samples of the test group. We were able to detect
\Mtest\ malware files using the malicious trees inferred from \Mtrain\ malware
files, and show that benign programs are benign, thus a \FPR\ false positive rate.  




\subsection{Inferring malicious behaviors} 

To infer malicious behaviors, we transformed each of the \Mtrain\ malware
binary files into a \PDS\ model using the approach described in Section
\ref{sec:inf}. To implement the oracle $\O$, we use the \emph{PoMMaDe} tool
\cite{Song:12a} that uses Jakstab \cite{Kinder:08} and IDA Pro
\cite{hex:2011}. Jakstab performs static analysis of the binary program.
However, it does not allow to extract API functions information, so IDA Pro is
used to obtain such information, thus obtaining $\varrho_{ar}$ and
$\varrho_{\lambda}$.  The $\varrho_{\tau}$ function was obtained by querying
the available information in the MSDN 
website.

We apply Algorithm \ref{alg:extract} 
to the \PDS\ models to extract \SCDT s for each of the malware instances. The current results were obtained with an $h$ value of 2. In
practice, to avoid the overapproximation of malicious trees, in the generation
of \SCDT s for the detection phase we consider the condition in line
\ref{algb:l8} of Algorithm \ref{alg:build}, $w_o[n]=w_d[m]$ true only when we
know the value outputted by the oracle is precise.  
%
%

\begin{wrapfigure}{r}{0.3\textwidth}
\vspace{-20pt}
\begin{tiny}
\centering
\begin{tabular}{lr}
\hline
Name  & \#\\
\hline
\hline
Backdoor.Win32.Agent & 26 \\
Worm.Win32.AutoRun & 13 \\
Email-Worm.Win32.Bagle & 19 \\
Email-Worm.Win32.Batzback & 4 \\
Backdoor.Win32.Bifrose & 46 \\
Backdoor.Win32.Hupigon & 5 \\
Email-Worm.Win32.Kelino & 7 \\
Trojan-PSW.Win32.LdPinch & 13 \\
Email-Worm.Win32.Mydoom & 26 \\
Email-Worm.Win32.Nihilit & 7 \\
Backdoor.Win32.SdBot & 14 \\
Backdoor.Win32.Small & 13 \\
\hline
Total & 193\\
\hline
\end{tabular}
\captionof{table}{Training dataset} 
\label{tab:trainf}
\end{tiny}
\vspace{-20pt}
\end{wrapfigure}

To compute
the \MSCDT\ we encode the extracted \SCDT\ as
graphs and try to calculate the most frequent subgraphs.  We use the gSpan
\cite{Yan:02} tool for that, it computes frequent subgraph structures using a
depth-first tree search over a canonical labeling of graph edges relying on the
linear ordering property of the labeling to prune the search space. The tool
has been applied in various domains as active chemical compound structure
mining and its performance is competitive among other tools \cite{Worlein:05}.
The tool supports only undirected graphs, therefore a mismatch with the trees
(that can be seen as rooted, acyclic direct graphs) used in this work. The
mismatch is overcome via a direction tag in the graph labels.

For the \Mtrain\ files extracted \SCDT s we have run the gSpan tool with
support 0.6\%. This is a tunable value for which we chose the one that allows better detection results. With this value we obtained 1026 subtrees (\MSCDT s), and best detection results.
From the inferred malicious trees output from gSpan, we build a 
 tree automaton  recognizing such trees.

 The training dataset contains 12 families of malware summarized in Table
\ref{tab:trainf}.
In average, our tool
extracts 7 \SCDT s in 30 seconds for each malware file. 
To store the 1026 discovered \MSCDT s the automaton file used 24Kb of memory.

\subsection{Detecting malware}

Malware detection is reduced to generating \SCDT s  and checking whether 
they are recognized by the inferred automaton.
Thus, to perform detection on an input binary file, we model it as \PDS\ using the
approach described in Section \ref{sec:inf} and extract \SCDT s using the
approach detailed in Section \ref{sec:computing}. If any subtree of the extracted 
tree is recognized by the automaton recognizing the malicious behaviors, 
we decide the binary sample is malware. We implemented such procedure in our 
tool and were able to detect \Mtest\ malware samples from \MtestF\ different families.

In  Table \ref{tab:testf} we show the range of malware families and number of
samples that our tool detects as malware. 
In average, our tool extracts 64 \SCDT s in 2.15 seconds for each
file (this value may be largely improved given that runtime efficiency was not a main goal of the prototype design). The discrepancy in the number of trees generated (compared to the
training set) is justified by an implementation choice regarding the oracle
approximation of the set of values pushed to the stack. In the generation of
\SCDT s for the detection phase we consider the condition in line \ref{algb:l8}
of Algorithm \ref{alg:build}, $w_o[n]=w_d[m]$ true even if the values are
approximated. Such cases were discarded in the generation of \SCDT s in the
inference step where it holds only when the oracle outputs precise values. The
automaton tree recognition execution time is negligible ($< 0.08$ secs) in all
cases.
To check the robustness of the detector, we applied it to a set of \Nbenware\
benign programs. Our tool was able to classify such programs as benign,
obtaining a \FPR\ false positive rate. In 88\% of the cases 
the tool extracts \SCDT s and at least in 44\% of the files 
there is a call to a function involved in malicious behavior (e.g.
\GetModuleFileName, \ShellExecute,\dots), but no tree was recognized as malicious. This value is in line with the values detailed in
\cite{Christodorescu:07,Fredrikson:10} and better than the 5\% reported in
\cite{Babic:11}.   

\begin{table}[htbp!]
\begin{minipage}{0.2425\textwidth}
\centering
\begin{tiny}
\begin{tabular}{lr}
\hline
Name  & \#\\
\hline
\hline
Backdoor.Win32.AF &        1 \\
Backdoor.Win32.Afbot &        1 \\
Backdoor.Win32.Afcore &        6 \\
Backdoor.Win32.Agent &       66 \\
Backdoor.Win32.Agobot &       47 \\
Backdoor.Win32.Alcodor &        1 \\
Backdoor.Win32.Antilam &        9 \\
Backdoor.Win32.Apdoor &        6 \\
Backdoor.Win32.Assasin &        3 \\
Backdoor.Win32.Asylum &        8 \\
Backdoor.Win32.Avstral &        2 \\
Backdoor.Win32.BLA &        2 \\
Backdoor.Win32.BNLite &        1 \\
Backdoor.Win32.BO2K &        6 \\
Backdoor.Win32.Bancodor &        1 \\
Backdoor.Win32.Bandok &        1 \\
Backdoor.Win32.Banito &        4 \\
Backdoor.Win32.Beastdoor &        6 \\
Backdoor.Win32.Bifrose &        5 \\
Backdoor.Win32.BoomRaster &        1 \\
Backdoor.Win32.Breplibot &        6 \\
Backdoor.Win32.Bushtrommel &        2 \\
Backdoor.Win32.ByShell &        1 \\
Backdoor.Win32.Cabrotor &        1 \\
Backdoor.Win32.Cafeini &        1 \\
Backdoor.Win32.Cheng &        1 \\
Backdoor.Win32.Cigivip &        1 \\
Backdoor.Win32.Cmjspy &        8 \\
Backdoor.Win32.Cocoazul &        2 \\
Backdoor.Win32.Codbot &        4 \\
Backdoor.Win32.Coldfusion &        3 \\
Backdoor.Win32.CommInet &        3 \\
Backdoor.Win32.Coredoor &        1 \\
Backdoor.Win32.Crunch &        1 \\
Backdoor.Win32.DKangel &        2 \\
Backdoor.Win32.DRA &        4 \\
Backdoor.Win32.DSNX &        3 \\
Backdoor.Win32.DarkFtp &        3 \\
Backdoor.Win32.DarkMoon &        1 \\
Backdoor.Win32.Delf &       31 \\
Backdoor.Win32.Dindang &        1 \\
Backdoor.Win32.DragonIrc &        1 \\
Backdoor.Win32.Dumador &        3 \\
Backdoor.Win32.Expir &        1 \\
Backdoor.Win32.HacDef &        2 \\
Backdoor.Win32.Hackarmy &        3 \\
Backdoor.Win32.Hupigon &        4 \\
Backdoor.Win32.IRCBot &        6 \\
Backdoor.Win32.Ierk &        1 \\
Backdoor.Win32.Jacktron &        1 \\
Backdoor.Win32.Jeemp &        1 \\
Backdoor.Win32.Katherdoor &        7 \\
Backdoor.Win32.Katien &        2 \\
Backdoor.Win32.Ketch &        4 \\
Backdoor.Win32.Kidterror &        1 \\
Backdoor.Win32.Konik &        1 \\
Backdoor.Win32.Krepper &        2 \\
Backdoor.Win32.Labrus &        1 \\
Backdoor.Win32.LanFiltrator &        2 \\
Backdoor.Win32.LanaFTP &        1 \\
Backdoor.Win32.Laocoon &        1 \\
Backdoor.Win32.Latinus &        5 \\
Backdoor.Win32.Lemerul &        1 \\
Backdoor.Win32.Lesbot &        1 \\
Backdoor.Win32.Levelone &        2 \\
Backdoor.Win32.Liondoor &        1 \\
Backdoor.Win32.Lithium &        3 \\
Backdoor.Win32.Litmus &        1 \\
Backdoor.Win32.LittleBusters &        1 \\
Backdoor.Win32.LittleWitch &        1 \\
Backdoor.Win32.Livup &        1 \\
Backdoor.Win32.Lixy &        1 \\
Backdoor.Win32.Lurker &        1 \\
Backdoor.Win32.Lyusane &        1 \\
Backdoor.Win32.MSNMaker &        1 \\
Backdoor.Win32.MServ &        1 \\
Backdoor.Win32.MainServer &        1 \\
Backdoor.Win32.Matrix &        3 \\
Backdoor.Win32.Medbot &        1 \\
Backdoor.Win32.Mellpon &        2 \\
Backdoor.Win32.Metarage &        1 \\
Backdoor.Win32.Mhtserv &        1 \\
Backdoor.Win32.Micronet &        1 \\
Backdoor.Win32.MiniCommander &        1 \\
\hline
\end{tabular}
\end{tiny}
\end{minipage}
\begin{minipage}{0.2425\textwidth}
\centering
\begin{tiny}
\begin{tabular}{lr}
\hline
Name  & \#\\
\hline
\hline
Backdoor.Win32.MoonPie &        1 \\
Backdoor.Win32.Mowalker &        1 \\
Backdoor.Win32.Mtexer &        2 \\
Backdoor.Win32.Mydons &        1 \\
Backdoor.Win32.Ncx &        1 \\
Backdoor.Win32.NerTe &        3 \\
Backdoor.Win32.NetControl &        2 \\
Backdoor.Win32.NetShadow &        1 \\
Backdoor.Win32.NetSpy &        8 \\
Backdoor.Win32.Netbus &        2 \\
Backdoor.Win32.Netdex &        2 \\
Backdoor.Win32.Netpocalipse &        1 \\
Backdoor.Win32.Neurotic &        2 \\
Backdoor.Win32.Nuclear &        3 \\
Backdoor.Win32.Nucledor &        2 \\
Backdoor.Win32.Nyrobot &        1 \\
Backdoor.Win32.Optix &        9 \\
Backdoor.Win32.PPCore &        1 \\
Backdoor.Win32.PPdoor &        2 \\
Backdoor.Win32.Pacak &        1 \\
Backdoor.Win32.Padodor &        5 \\
Backdoor.Win32.PcClient &       12 \\
Backdoor.Win32.PeepViewer &        1 \\
Backdoor.Win32.Peers &        2 \\
Backdoor.Win32.Penrox &        1 \\
Backdoor.Win32.Pepbot &        1 \\
Backdoor.Win32.Pingdoor &        1 \\
Backdoor.Win32.Pipes &        1 \\
Backdoor.Win32.Plunix &        1 \\
Backdoor.Win32.Pornu &        1 \\
Backdoor.Win32.Probot &        1 \\
Backdoor.Win32.Proxydor &        2 \\
Backdoor.Win32.Psychward &        5 \\
Backdoor.Win32.Ptakks &        1 \\
Backdoor.Win32.Puddy &        1 \\
Backdoor.Win32.R3C &        1 \\
Backdoor.Win32.RAT &        2 \\
Backdoor.Win32.RDR &        1 \\
Backdoor.Win32.Rbot &        8 \\
Backdoor.Win32.Redkod &        4 \\
Backdoor.Win32.Revenge &        1 \\
Backdoor.Win32.Rirc &        1 \\
Backdoor.Win32.Robobot &        1 \\
Backdoor.Win32.Ronater &        1 \\
Backdoor.Win32.Rootcip &        1 \\
Backdoor.Win32.Roron &        1 \\
Backdoor.Win32.RtKit &        4 \\
Backdoor.Win32.Ruledor &        4 \\
Backdoor.Win32.SPing &        3 \\
Backdoor.Win32.SatanCrew &        1 \\
Backdoor.Win32.Sbot &        2 \\
Backdoor.Win32.SdBot &       63 \\
Backdoor.Win32.Seed &        3 \\
Backdoor.Win32.Serman &        1 \\
Backdoor.Win32.ShBot &        1 \\
Backdoor.Win32.Shakdos &        1 \\
Backdoor.Win32.Shox &        1 \\
Backdoor.Win32.SilverFTP &        1 \\
Backdoor.Win32.Sinf &        1 \\
Backdoor.Win32.Sinit &        4 \\
Backdoor.Win32.SkyDance &        1 \\
Backdoor.Win32.Small &       22 \\
Backdoor.Win32.Sporkbot &        1 \\
Backdoor.Win32.SpyBoter &        9 \\
Backdoor.Win32.Stang &        1 \\
Backdoor.Win32.Stats &        1 \\
Backdoor.Win32.Stigmador &        1 \\
Backdoor.Win32.SubSeven &        1 \\
Backdoor.Win32.Sumatrix &        1 \\
Backdoor.Win32.Suslix &        1 \\
Backdoor.Win32.Symes &        1 \\
Backdoor.Win32.Sysinst &        1 \\
Backdoor.Win32.System33 &        1 \\
Backdoor.Win32.Sytr &        1 \\
Backdoor.Win32.TDS &        3 \\
Backdoor.Win32.Takit &        1 \\
Backdoor.Win32.Tasmer &        1 \\
Backdoor.Win32.Telemot &        1 \\
Backdoor.Win32.TheThing &        3 \\
Backdoor.Win32.Thunk &        1 \\
Backdoor.Win32.Tonerok &        3 \\
Backdoor.Win32.URCS &        2 \\
Backdoor.Win32.Undernet &        1 \\
Backdoor.Win32.Unwind &        1 \\
\hline
\end{tabular}
\end{tiny}
\end{minipage}
\begin{minipage}{0.2425\textwidth}
\centering
\begin{tiny}
\begin{tabular}{lr}
\hline
Name  & \#\\
\hline
\hline
Backdoor.Win32.UpRootKit &        1 \\
Backdoor.Win32.Ursus &        1 \\
Backdoor.Win32.Utilma &        1 \\
Backdoor.Win32.VB &        2 \\
Backdoor.Win32.VHM &        1 \\
Backdoor.Win32.Vatos &        1 \\
Backdoor.Win32.Verify &        1 \\
Backdoor.Win32.WMFA &        1 \\
Backdoor.Win32.WRT &        1 \\
Backdoor.Win32.WbeCheck &        3 \\
Backdoor.Win32.Webdor &        6 \\
Backdoor.Win32.Whisper &        1 \\
Backdoor.Win32.Wilba &        1 \\
Backdoor.Win32.Winker &        5 \\
Backdoor.Win32.WinterLove &        7 \\
Backdoor.Win32.Wisdoor &        7 \\
Backdoor.Win32.Wollf &        4 \\
Backdoor.Win32.XBot &        1 \\
Backdoor.Win32.XConsole &        1 \\
Backdoor.Win32.XLog &        2 \\
Backdoor.Win32.Xdoor &        2 \\
Backdoor.Win32.Y2KCount &        1 \\
Backdoor.Win32.Ythac &        1 \\
Backdoor.Win32.Zerg &        1 \\
Backdoor.Win32.Zombam &        1 \\
Backdoor.Win32.Zomby &        1 \\
Constructor.Win32.Delf &        1 \\
Constructor.Win32.ETVM &        2 \\
Constructor.Win32.EvilTool &        1 \\
Constructor.Win32.MS04-032 &        1 \\
Constructor.Win32.MS05-009 &        1 \\
Constructor.Win32.SPL &        1 \\
Constructor.Win32.SS &        2 \\
Constructor.Win32.VCL &        1 \\
DoS.Win32.Aspcode &        1 \\
DoS.Win32.Ataker &        1 \\
DoS.Win32.DStorm &        1 \\
DoS.Win32.Igemper &        1 \\
DoS.Win32.SQLStorm &        1 \\
Email-Worm.Win32.Anar &        2 \\
Email-Worm.Win32.Android &        1 \\
Email-Worm.Win32.Animan &        1 \\
Email-Worm.Win32.Anpir &        1 \\
Email-Worm.Win32.Ardurk &        2 \\
Email-Worm.Win32.Asid &        1 \\
Email-Worm.Win32.Assarm &        1 \\
Email-Worm.Win32.Atak &        1 \\
Email-Worm.Win32.Avron &        2 \\
Email-Worm.Win32.Bagle &        3 \\
Email-Worm.Win32.Bagz &        5 \\
Email-Worm.Win32.Banof &        1 \\
Email-Worm.Win32.Bater &        1 \\
Email-Worm.Win32.Batzback &        3 \\
Email-Worm.Win32.Blebla &        1 \\
Email-Worm.Win32.Bumdoc &        2 \\
Email-Worm.Win32.Charch &        1 \\
Email-Worm.Win32.Cholera &        1 \\
Email-Worm.Win32.Coronex &        3 \\
Email-Worm.Win32.Cult &        1 \\
Email-Worm.Win32.Delf &        4 \\
Email-Worm.Win32.Desos &        1 \\
Email-Worm.Win32.Donghe &        3 \\
Email-Worm.Win32.Drefir &        1 \\
Email-Worm.Win32.Duksten &        2 \\
Email-Worm.Win32.Dumaru &       10 \\
Email-Worm.Win32.Energy &        1 \\
Email-Worm.Win32.Entangle &        1 \\
Email-Worm.Win32.Epon &        1 \\
Email-Worm.Win32.Eyeveg &        3 \\
Email-Worm.Win32.Fix2001 &        1 \\
Email-Worm.Win32.Frethem &        2 \\
Email-Worm.Win32.Frubee &        1 \\
Email-Worm.Win32.GOPworm &        1 \\
Email-Worm.Win32.Gift &        2 \\
Email-Worm.Win32.Gismor &        1 \\
Email-Worm.Win32.Gizer &        2 \\
Email-Worm.Win32.Gunsan &        2 \\
Email-Worm.Win32.Haltura &        1 \\
Email-Worm.Win32.Hanged &        1 \\
Email-Worm.Win32.Happy &        1 \\
Email-Worm.Win32.Ivalid &        1 \\
Email-Worm.Win32.Jeans &        1 \\
Email-Worm.Win32.Kadra &        1 \\
Email-Worm.Win32.Keco &        3 \\
\hline
\end{tabular}
\end{tiny}
\end{minipage}
\begin{minipage}{0.2425\textwidth}
\centering
\begin{tiny}
\begin{tabular}{lr}
\hline
Name  & \#\\
\hline
\hline
Email-Worm.Win32.Kelino &        6 \\
Email-Worm.Win32.Kergez &        1 \\
Email-Worm.Win32.Kipis &        2 \\
Email-Worm.Win32.Kirbster &        1 \\
Email-Worm.Win32.Klez &        9 \\
Email-Worm.Win32.Lacrow &        2 \\
Email-Worm.Win32.Lara &        1 \\
Email-Worm.Win32.Lentin &       10 \\
Email-Worm.Win32.Locksky &        2 \\
Email-Worm.Win32.Lohack &        3 \\
Email-Worm.Win32.LovGate &        3 \\
Email-Worm.Win32.Mescan &        1 \\
Email-Worm.Win32.Mimail &        1 \\
Email-Worm.Win32.Miti &        1 \\
Email-Worm.Win32.Modnar &        1 \\
Email-Worm.Win32.Mydoom &        8 \\
Email-Worm.Win32.NWWF &        1 \\
Email-Worm.Win32.Navidad &        1 \\
Email-Worm.Win32.NetSky &        2 \\
Email-Worm.Win32.NetSup &        1 \\
Email-Worm.Win32.Netav &        1 \\
Email-Worm.Win32.Newapt &        6 \\
Email-Worm.Win32.Nihilit &        1 \\
Email-Worm.Win32.Nirky &        1 \\
Email-Worm.Win32.Paroc &        1 \\
Email-Worm.Win32.Parrot &        1 \\
Email-Worm.Win32.Pepex &        2 \\
Email-Worm.Win32.Pikis &        2 \\
Email-Worm.Win32.Plage &        1 \\
Email-Worm.Win32.Plexus &        1 \\
Email-Worm.Win32.Pnguin &        1 \\
Email-Worm.Win32.Poo &        1 \\
Email-Worm.Win32.Postman &        1 \\
Email-Worm.Win32.Qizy &        1 \\
Email-Worm.Win32.Rammer &        1 \\
Email-Worm.Win32.Rapita &        1 \\
Email-Worm.Win32.Rayman &        1 \\
Email-Worm.Win32.Repah &        2 \\
Email-Worm.Win32.Ronoper &       20 \\
Email-Worm.Win32.Roron &       23 \\
Email-Worm.Win32.Sabak &        1 \\
Email-Worm.Win32.Savage &        2 \\
Email-Worm.Win32.Scaline &        1 \\
Email-Worm.Win32.Scrambler &        1 \\
Email-Worm.Win32.Seliz &        1 \\
Email-Worm.Win32.Sharpei &        1 \\
Email-Worm.Win32.Silly &        1 \\
Email-Worm.Win32.Sircam &        1 \\
Email-Worm.Win32.Skudex &        2 \\
Email-Worm.Win32.Sonic &        4 \\
Email-Worm.Win32.Stator &        1 \\
Email-Worm.Win32.Stopin &        3 \\
Email-Worm.Win32.Sunder &        1 \\
Email-Worm.Win32.Svoy &        2 \\
Email-Worm.Win32.Swen &        1 \\
Email-Worm.Win32.Tanatos &        3 \\
Email-Worm.Win32.Taripox &        2 \\
Email-Worm.Win32.Totilix &        1 \\
Email-Worm.Win32.Trilissa &        4 \\
Email-Worm.Win32.Trood &        2 \\
Email-Worm.Win32.Unis &        1 \\
Email-Worm.Win32.Urbe &        3 \\
Email-Worm.Win32.Valha &        1 \\
Email-Worm.Win32.Volag &        1 \\
Email-Worm.Win32.Vorgon &        2 \\
Email-Worm.Win32.Warezov &        1 \\
Email-Worm.Win32.Winevar &        1 \\
Email-Worm.Win32.Wozer &        1 \\
Email-Worm.Win32.Xanax &        2 \\
Email-Worm.Win32.Yanz &        1 \\
Email-Worm.Win32.Yenik &        1 \\
Email-Worm.Win32.Zircon &        4 \\
Exploit.Win32.Agent &        3 \\
Exploit.Win32.AntiRAR &        1 \\
Exploit.Win32.CAN &        1 \\
Exploit.Win32.CVE-2006-1359 &        1 \\
Exploit.Win32.CrobFTP &        1 \\
Exploit.Win32.DCom &        3 \\
Exploit.Win32.DameWare &        1 \\
Net-Worm.Win32.Muma &        1 \\
Trojan-PSW.Win32.LdPinch &       16 \\
Worm.Win32.AutoRun &       34 \\
\\
\hline
Total & 983\\
\hline
\end{tabular}
\end{tiny}
\end{minipage}
\caption{Test dataset name family and number of samples (\#) detected}
\label{tab:testf}
\end{table}

\section{Related work}
\label{sec:rel}

Malicious behaviors have been defined in different ways. The foundational
approaches via computable functions \cite{Adleman:88},  based in  Kleene's
recursion theorem \cite{Bonfante:05,Bonfante:06,Bonfante:07}, or the neat
definition using MALog \cite{Kramer:10} capture the essence of such behaviors,
but are too abstract to be used in practice or require the full specification
of software functionality.
Our work is close to the approaches using model checking and temporal logic
formulas as malicious behavior specification \cite{Song:12b,Song:12a}. In such 
works specifications have to be designed by hand while we are able to learn 
them automatically. Some of the trees we infer describe malicious 
behaviors encoded in such formulas. 

Regarding semantic signature inference there are the works 
\cite{Christodorescu:07,Fredrikson:10} where the extraction of behaviors is
based on dynamic analysis of executables. From the execution traces collected,
data flow dependencies among system calls are recovered by comparing parameters
and type information. The outcome are dependence graphs where the nodes are
labeled by system function names and  the edges capture the dependencies
between the system calls. 
Another dynamic analysis based approach is the one of \cite{Babic:11} where trees, 
alike ours, express the same kind of data flows between 
nodes representing system calls.
Both approaches are limited by the drawbacks of dynamic analysis. For instance, 
time limitations,  
 limited
system call tracing or an overhead up to $90\times$ slower during execution
\cite{Skaletsky:10}. Plus, from the dataset made publicly available in
\cite{Babic:11}, we notice the signatures involve only functions from the Native
API library. Our approach has the advantage of being API independent, thus
the level of analysis may be tuned, plus Win32 API function based signatures
should be shorter as each high level function should be translated into a set of
calls to the Native API functions.

In \cite{Bonfante:09} the authors propose to learn behaviors of binary files 
by extracting program control-flow graphs using dynamic analysis. Such graphs 
contain assembly instructions that correspond to control flow information e.g. $jmp$, but 
that introduces more possibilities to circumvent such signatures by rewriting 
the code. From the graphs, trees are computed and the union of all such trees 
is used to infer an automaton that is used in detection. Our inference does
not output all the trees, only the most frequent, improving the learning 
process and generalizing from the training dataset.  

An alternative to semantic signatures are works based on machine
learning  approaches as \cite{Tahan:12}, which shows that by mining ``$n-$grams''
(a sequence of $n$ bits), it is possible to distinguish malware from benign program.  
In our approach, the distinguishing features (malicious behaviors) can be seen as 
traces of program execution, thus having a meaning that can be more easily understood.


\section{Conclusion}
\label{sec:conc}
In this work, we have shown how to combine {\em static} reachability analysis techniques to infer
malware semantic signatures in the form of malicious trees, which describe
the data flows among system calls. 
Our experiments show
that the approach can be used to automatically infer specifications of
malicious behaviors and detect several malware samples from an a priori given
smaller set of malware.  We were able to detect
\Mtest\ malware files using the malicious trees inferred from \Mtrain\ malware
files, and applied the detector to \Nbenware\ benign files obtaining a \FPR\ false positive rate.





As future work we envisage the improvement of the binary modeling techniques,
for example enriching the function parameter type system to allow better
approximations.  The usage of more advanced  mining techniques, e.g.\
structural leap mining used in \cite{Fredrikson:10}, can be used to improve the
learning approach.  In another direction, given the relation between modal formulas and tree models a
comparison between our approach and the approach in \cite{Song:12b} concerning
expressiveness and complexity is envisaged.  Finally, a complexity study with
respect to the depth of the trees extraction (parameter $h$ in Algorithm
\ref{alg:extract}) and size of the \HELTA\  would be another alternative
direction.


Summing up, the reachability analysis of \PDS\ models of executables can play a major role in the
malware specification inference domain. 
The ability to precisely analyze stack
behavior enables the extraction of executables system call
data flows and overcomes typical obfuscated calls to such routines.



\end{document}